\documentclass[a4paper,11pt]{article}
\usepackage{amssymb,latexsym,amsmath,amsxtra,amsfonts,color}
\usepackage{graphicx,bbm}
\usepackage{parskip}
\usepackage{mathrsfs}

\newtheorem{thm}{Theorem}[section]
\newtheorem{lem}[thm]{Lemma}
\newtheorem{prop}[thm]{Proposition}
\newtheorem{corr}[thm]{Corollary}

\newtheorem{defn}[thm]{Definition}

\newtheorem{remark}[thm]{Remark}

\newenvironment{proof}{{\bf Proof.}  }{\hfill$\blacksquare$}

\def\R{\mathbb{R}}  
\def\C{\mathbb{C}}
\newcommand{\at}[2]{\left.#1\right|_{#2}} 
\newcommand{\D}[2]{\frac{d#1}{d#2}} 
\newcommand{\DO}[1]{\frac{d}{d#1}} 
\newcommand{\PD}[2]{\frac{\partial#1}{\partial#2}} 
 
\newcommand{\PDDM}[3]{\frac{\partial^{2}{#1}}{\partial{#2}\partial{#3}}}
\newcommand{\norm}[1]{\Vert#1\Vert}
\newcommand{\abs}[1]{\left|#1\right|}
\newcommand{\paren}[1]{\left(#1\right)}
\newcommand{\brac}[1]{\left[#1\right]}
\newcommand{\dpair}[2]{\left<#1,#2\right>}
\newcommand{\ind}{\mathbf{1}}
\newcommand{\vx}{\boldsymbol{x}}
\newcommand{\vg}{\boldsymbol{\gamma}}
\newcommand{\vy}{\boldsymbol{y}}
\newcommand{\vl}{\boldsymbol{L}}
\newcommand{\vm}{\boldsymbol{M}}
\newcommand{\til}[1]{\widetilde{#1}}

\begin{document}

\title{Emergence of Balance from a model of Social Dynamics.}

\author{Ikemefuna Agbanusi \\ and\\Jared C. Bronski\\ University of Illinois Department of Mathematics\\ 1409 W Green St. \\ Urbana, IL 61801.} 

\date{\today.}

\maketitle

\begin{abstract}
We propose a model for social dynamics on a network. In this model each actor 
holds a position on some issue, actors and their opinions being associated 
to vertices of the graph, and,
additionally, the actors hold opinions of one another, with these opinions 
being associated to edges in the graph. These quantities are
allowed to evolve according to the gradient flow of a natural free energy.
We show that for a small spread in opinions the model converges to 
a consensus state, where all actors hold the same position. For a larger 
spread in opinion there is a phase transition marked by the birth of a second 
stable state: in addition to the consensus state there is a second polarized 
or partisan state. This state, when it exists, is conjectured to be global 
energy minimizer, with the consensus state being a local energy minimizer. 
We derive an energy inequality which supports, though does not prove, this. 
Interestingly, all of the steady 
states we find, with the exception of the consensus state, are either balanced 
(in the sense of Heider) or are completely unbalanced states where  all 
triangles are unbalanced. The latter solutions are, not surprisingly, always 
unstable.
\end{abstract}

\section{Introduction}
There has been great interest lately in developing mathematical models to 
understand emergent social phenomenon. Some of the many models considered 
include spin-like models for
opinion dynamics~\cite{Galam.1990,Mobilia.2003} and cultural dynamics~\cite{Axelrod.1997,Castellano.2000}.  
In this paper we introduce a model for the co-evolution of opinions and
positions in a social network in order to understand the dynamics of 
balance. The idea of balance dates to the work of 
Heider~\cite{Heider.1946}, who argued that in a stable network of 
relationships every triad should have an even number of negative (antagonistic)
edges.  In essence these networks are ones which satisfy the aphorism 
``the enemy of my enemy is my friend.''
For example balance theory suggests that a network with three mutually 
antagonistic groups is unstable, with two of the groups making common cause against the third.
Harary and Cartwright~\cite{Cartwright.Harary.1956} generalized 
this condition to require an even number of negative 
edges in every cycle, and showed that balanced graphs are exactly 
bi-partitite graphs where edges within a group have positive weights and 
edges between groups have a negative weight. 

The works of Heider and Cartwright and Harary are static, but the language 
used is strongly suggestive of a dynamical process, and there have been 
several attempts to introduce dynamical models of this process. 
One such model was introduced by Antal, Krapivsky and 
Redner~\cite{Antal.Krapivsky.Redner.2005}. 
In this discrete time model, unstable triangles transition, with some 
probability, to stable ones by flipping the sign of an edge. Another model 
is the one introduced by Kulakowski, Gawr\'onski and Gronek,~\cite{Kulakowski.Gawronski.Gronek.2005} and 
later analyzed by Marvel, Kleinberg, Kleinberg and Strogatz~\cite{Marvel.Kleinberg.Kleinberg.Strogatz.2010,Marvel.Strogatz.Kleinberg.2009}. 
The Kulakowski-Gawr\'onski-Gronek model takes the form of a single matrix 
Ricatti equation 
\[
\D{{\bf X}}{t} = {\bf X}^2,
\]
where $X_{ij}$ denotes the opinion actor $i$ holds of actor $j$. This model 
was explicitly solved in the symmetric case ($X_{ij}=X_{ji}$) by  Marvel, 
Kleinberg, Kleinberg and Strogatz, who showed that for generic initial 
conditions the matrix ${\bf X}(t)$ converges in finite time to a balanced state. 

While these models are very interesting, they are phenomenonological: 
the form of the dynamics is chosen to drive the edge weights
towards the balanced state. It would be preferable to find a model 
in which the balance state was not assumed but rather 
emergent from the dynamics. Further in the modeling it is not necessarily 
desirable to assume that the underlying graph is the complete graph, where all
actors know each other, but it is not clear how to extend the  Kulakowski-Gawr\'onski-Gronek model to a more general graph which might have few or no triangles.

In this paper we propose and
analyze such a model: each actor behaves in a very natural way, and balance 
arises naturally from the asymptotic steady states of the model. 

The model we consider is posed on a graph $\Gamma$, representing  a network 
of relations. The graph has $N$ vertices, representing a number of actors, and $\abs{E}$ edges, 
representing the relations between pairs of actors. For this paper we assume 
that the underlying graph is the complete graph, where all actors know each 
other, so $\abs{E}={N\choose{2}}$ but the model extends in 
a straightforward way to an arbitrary graph. There are 
two types of variables in this model: 
\begin{itemize}
\item Positions $x_i(t)$ are associated with the vertices, and represent the 
position of actor $i$ on some issue which can be represented as a continuum:
conservative vs. liberal, tastes great vs. less filling, etc. 
\item Opinions $\gamma_{ij}$ are associated with the edges in the 
graph, and  represent the degree of friendliness or respect between actor 
$i$ and actor $j$, with $\gamma_{ij}>0$ representing friendliness and 
$\gamma_{ij}<0$ antagonism. 
\end{itemize}

We will not initially assume that the opinions $\gamma_{ij}$ are symmetric: 
but we will show that this emerges naturally 
from the dynamics: the steady states of the model all have the property 
that the opinions are symmetric: $\gamma_{ij}=\gamma_{ji}$.

We associate to the quantities $x_i,\gamma_{ij}$ a  Dirichlet energy 
\[
{\mathcal D}(\vx,\vg)=\sum_{i>j} \gamma_{ij} (x_i-x_j)^2,
\]
which represents the total amount of disharmony in the system. Note that 
$\gamma_{ij}$ above can be of either sign. If all $\gamma_{ij}>0$  
(friendly relations) then the energy is minimized when the actors take the same 
position, $x_i=x_j$: friends like to agree. If, on the other 
hand, $\gamma_{ij}<0$, then the energy is minimized when $(x_i-x_j)^2$ is large: 
antagonists prefer to disagree. 

The basic dynamics of the model is as follows: we assume that all actors act 
continuously in time so as to minimize ${\mathcal D}(\vec x,\vec \gamma)$ subject 
to the following constraints.
\begin{align}
\frac{1}{2E} \sum_{i\neq j} \gamma_{ij} &= Q>0\label{eqn:gam_sum_cons}\\
\frac{1}{2E} \sum_{i\neq j} \gamma_{ij}^2 &=P\label{eqn:gam_norm_cons} \\
\sum_{i=1}^N x_i^2 &=R \label{eqn:x_norm_cons}
\end{align}
The first constraint requires that the sum of the opinions must be a positive 
constant. This can be interpreted as a societal pressure towards civil 
discourse: while actors may hold negative opinions of each other the 
average opinion must be positive. The second constraint guarantees that no 
actor can hold an opinion that is too extreme. Note that the Cauchy-Schwartz inequality implies that
\[
P-Q^2\geq 0.
\]  
The quantity $P-Q^2$ represents some socially acceptable range of opinions, 
and thus is analogous to an entropy. The Lagrange multiplier that enforces 
this constraint ($\tau$, defined below) can therefore be thought of as being
like a temperature. Finally the third constraint guarantees that none of the positions are too extreme. 

The positions $x_i$ and the opinions 
$\gamma_{ij}$ evolve according to a constrained gradient flow. Following the 
method of Lagrange multipliers the free energy is given by 

\begin{equation}
\mathscr{D}:= \frac{1}{2}\sum_{i\neq j}\gamma_{ij}(x_i-x_j)^2-
\frac{1}{2}\frac{\mu}{\abs{E}}\sum_{i\neq j}\gamma_{ij}-\frac{1}{2}\frac{\tau}{\abs{E}}\sum_{i\neq j}\gamma^2_{ij}-\lambda\sum_i x_i^2,
\end{equation}

where $\tau,\mu,\lambda$ are the three Lagrange multipliers enforcing the 
constraints \eqref{eqn:gam_sum_cons}-\eqref{eqn:x_norm_cons}. The equations of motion for 
$x_i$ and $\gamma_{ij}$ are given by 
\begin{equation}
\left.\begin{aligned}
 \dot{x}_i&=-\PD{\mathscr{D}}{x_i}\\
 \dot{\gamma}_{ij}&=-\epsilon\PD{\mathscr{D}}{\gamma_{ij}}
\end{aligned}
\right\}
\end{equation}
or, more explicitly,
\begin{align}
 \dot x_i &= -2\paren{\frac{1}{2}\sum_{j\neq i}\gamma_{ij}(x_i-x_j) - \lambda x_i} && i\in (1\ldots N)\label{eqn:posn}\\ 
 \dot \gamma_{ij} &= -\epsilon\paren{\frac12(x_i - x_j)^2 - \frac{\tau}{\abs{E}} \gamma_{ij} - \frac{\mu}{2\abs{E}}} && i,j\in(1\ldots N) \label{eqn:opin}.
\end{align}
Here we have introduced a ``stiffness'' parameter, $\epsilon$, which measures the ease
with which actors change their opinions of one another.

The Lagrange multipliers are dynamic quantities, and are determined 
by the conditions that $P,Q,R$ be constant. For example
\[
 0=\dot{R} =2\sum_{i=1}^Nx_i\dot{x}_i
\]
from which we get that
\begin{equation}\label{eqn:lambda}
 \lambda =\frac{\dpair{\vx}{\vl\vx}}{\norm{\vx}^2}.
\end{equation}
Here ${\vl}$ is the graph Laplacian given by
\begin{equation}\label{eqn:graph_Laplacian}
  {\vl}_{ij} = {\vl}_{ij}(\vec\gamma)=
 \begin{cases}
 -\frac12(\gamma_{ij}+\gamma_{ji}), & i\neq j,\\
  \frac12\sum_{k \neq i}(\gamma_{ik}+\gamma_{ki}),& i = j.
\end{cases}
\end{equation}

We pause to discuss the physical interpretation of the dynamics. 
Eqn \eqref{eqn:posn} is a nonlinear ($\lambda$ depends on ${\vec x}$ 
as through (\ref{eqn:lambda})) heat flow representing the relaxation to 
consensus. Actors adjust their positions $x_i$ towards the 
positions of those that the actor respects ($\gamma_{ij}>0$) and away from 
the positions of those that the actor does not respect ($\gamma_{ij}<0$). 
We also note that models similar to the $x$ evolution have been 
previously been considered in physical applications to social sciences. 
In many of these models $x$ is an Ising-like spin variable, representing the 
choice between two options, rather than a continuous variable, 
but the general flavor is similar. For an introduction to the extensive 
literature on these models we refer the interested reader to the papers of 
Durlauf~\cite{Durlauf.1999}, Galam~\cite{Galam.2008}, Castellano,
Fortunato and Loreta~\cite{Castellano.Fortunato.Loreto.2009}, Lim~\cite{Lim.2011},
and Shi, Mucha and Durrett\cite{Shi.Mucha.Durrett.2013} .

The second equation \eqref{eqn:opin} reflects the tendency of actors to adjust their
opinions of other actors in response to relative differences in their positions. 
The first term on the right-hand side above represents the squared difference in the 
positions of the two actors, while the remaining two terms represent an average difference in 
opinion over the whole network. If the two actors hold positions that are 
close, relative to the average spread in position over the whole network, the 
opinion the actors hold of each other goes up ($\gamma_{ij}$ increases), while 
if their positions are relatively far apart, the opinion they hold of each 
other goes down.  

\section{Preliminaries}
For convenience we start by defining
\begin{align*}
&g_1(\vx,\vg)=\frac{1}{2\abs{E}}\sum_{i\neq j}\gamma_{ij}\\
&g_2(\vx,\vg)=\frac{1}{2\abs{E}}\sum_{i\neq j}\gamma^2_{ij}\\
&g_3(\vx,\vg)=\sum_{i=1}^N x_i^2
\end{align*}
 so that the constraints \eqref{eqn:gam_sum_cons}, \eqref{eqn:gam_norm_cons}, 
\eqref{eqn:x_norm_cons} are equivalent to $g_1(\vx,\vg)=Q$, $g_2(\vx,\vg)=P$ and $g_3(\vx,\vg)=R$
 respectively and 
 \[
\mathscr{D}(\vx,\vg)=\mathcal{D}-\mu g_1-\tau g_2-\lambda g_3.
\]
Defining
\[
\widetilde{\Omega} = \{(\vx,\vg)\in\R^{N}\times\R^{2|E|}:g_1=Q;\, g_2=P;\, g_3=R\},
\]
we see that the gradient flow is constrained to the set $\til{\Omega}$.
If $P>Q^2$, the sphere defined by $g_2(\vx,\vg)=0$ and the hyperplane defined by $g_1(\vx,\vg)=0$ 
intersect transversely.  An application of the implicit function theorem now shows that $\til{\Omega}$ is actually a 
smooth compact manifold of codimension $3$.  

Next we observe
\begin{prop}
For $\epsilon>0$, the model always tends to a state in which the opinions 
are symmetric, i.e. $\gamma_{ij}=\gamma_{ji}$.
\end{prop}
To see this we note that 
the model is a gradient flow on the compact set $\til{\Omega}$ defined above, and thus always tends to 
a local energy minimizer i.e. a critical point of $\mathscr{D}$. From \eqref{eqn:opin} we see that the symmetric 
difference in the opinions $s_{ij}:=\gamma_{ij}-\gamma_{ji}$ satisfies
\[
\dot{s}_{ij} = \frac{2\epsilon\tau}{\abs{E}}s_{ij}.
\]
Since the compactness of $\til{\Omega}$ forbids exponential growth,we get that $\gamma_{ij}-\gamma_{ji}$ tends to zero asymptotically.
This in turn implies that $\tau<0$ at a stable fixed point. Since the model always tends to a state in which the opinions are
symmetric, we will, for the remainder of the paper, assume that $\gamma_{ij} = \gamma_{ji}$.

For the reader's convenience, we rewrite the equations under the symmetry assumption:
\begin{align*}
&{\mathcal D}(\vx,\vg)=\sum_{i< j} \gamma_{ij} (x_i-x_j)^2,\\
&g_1=\frac{1}{\abs{E}}\sum_{i<j}\gamma_{ij}\\
&g_2=\frac{1}{\abs{E}}\sum_{i< j}\gamma^2_{ij}\\
&g_3=\sum_{i=1}^N x_i^2
\end{align*}
and
\[
 \Omega = \{(\vx,\vg)\in\R^{N}\times\R^{|E|}:g_1=Q;\, g_2=P;\, g_3=R\},
\]
which is an $N+\abs{E}-3$ dimensional manifold.

As in the Introduction we derive expressions for the Lagrange multipliers $\mu$ and $\tau$
using the fact that $\dot{Q}=0$ and $\dot{R}=0$.
Since 
\[
\dot{\gamma}_{ij} =-\epsilon\brac{(x_i-x_j)^2-\frac{\mu}{\abs{E}}-\frac{2\tau}{\abs{E}}\gamma_{ij}},
\] 
we see that $\dot{Q}=0$ implies
\begin{align*}
{\sum_{i< j}(x_i-x_j)^2-\mu-2\tau Q}=0,
\end{align*}
and $\dot{R}=0$ implies
\begin{align*}
{\sum_{i< j}\gamma_{ij}(x_i-x_j)^2-\mu Q-2\tau P}=0.
\end{align*}
 Solving for $\mu$ and $\tau$ gives
 \begin{align}
 \tau&= \frac{1}{2(Q^2-P)}\brac{Q\sum_{i< j}(x_i-x_j)^2-\sum_{i<j}\gamma_{ij}(x_i-x_j)^2}\label{eqn:tau}\\
 \mu& = \frac{1}{(P-Q^2)}\brac{P\sum_{i<j}(x_i-x_j)^2-Q\sum_{i< j}\gamma_{ij}(x_i-x_j)^2}\label{eqn:mu}
 \end{align}

For an arbitrary complex matrix A, we let $\sigma(A)$ denote its spectrum. If $\sigma(A)\subset\R$, we also denote by 
$\sigma_{max}(A)$ and $\sigma_{min}(A)$ its largest and smallest eigenvalues respectively.

\section{Fixed Points of the Model}
Since we know that the model will (generically) converge to local energy 
minimizers we look at the possible fixed points of the flow. These
are given by the solutions to the equations 
\begin{align*}
{\vl} {\vx} &= \lambda \vx \\
\frac{2\tau}{|E|}\gamma_{ij} + \frac{\mu}{|E|} &= (x_i-x_j)^2.  
\end{align*}
The fixed points represent critical points of the free energy, but they 
are not necessarily energy minimizers and may otherwise represent critical 
points of the energy corresponding to unstable equilibria.  

These are simultaneous polynomial equations, so in general it is difficult to 
find all solutions, but we have been able to find a number of exact solutions 
corresponding to all of the observed behaviors in the system. 
 
We begin by noting that the vector $\vx = (1,1,1\ldots,1)^t$
is always in the null-space of ${\vl}$ and thus is always a fixed 
point of the model regardless of the opinions 
$\gamma_{ij}.$ This makes sense: if all actors hold the same position there is 
nothing to drive the conflict. Thus we refer to this as the \emph{consensus state}.
Technically speaking the consensus state is not a critical \emph{point} since, as we have
observed, it exists for all values of $\vg$:

\begin{defn}
The consensus state is the (critical) manifold $\mathcal{C}$ defined by
\[\mathcal{C}:=\{(\vx,\vg)\in\Omega: \vx= \sqrt{R/N}(1,1,1\ldots,1)^t\}.\]
\end{defn}

To describe the other set of fixed points we introduce some notation.
Let $I,J$ be subsets of the vertex set $V(\Gamma)$ such that
$I\cap J=\varnothing$ and $I\cup J=V(\Gamma)$. Consider the matrix
\begin{equation}\label{eqn:bipartite}
M_{ij}=
\begin{cases}
-\alpha;& i\neq j\text{\, in same subset}\\
-\beta;& i\neq j\text{\, in different subset}\\
(\abs{I}-1)\alpha+\abs{J}\beta; & i=j\in I\\
(\abs{J}-1)\alpha+ \abs{I}\beta; &i=j\in J.
\end{cases}
\end{equation}
Note that $M$, defined above, is symmetric and also has row and column sum to zero.  We then have the following 
characterization if its spectrum:

\begin{lem} The spectrum of $M$, defined by \eqref{eqn:bipartite}, is given by
\begin{equation}
  \sigma(M)=
 \begin{cases} 
 0,&\text{\,\, of multiplicity}=1\\
 N\beta,&\text{\,\, of multiplicity}=1\\
 \abs{I}\alpha+\abs{J}\beta,&\text{\,\, of multiplicity}=\abs{I}-1\\
  \abs{I}\beta+\abs{J}\alpha,&\text{\,\, of multiplicity}=\abs{J}-1\\
 \end{cases}
 \end{equation}
 \end{lem}
 
 \begin{proof}
 The proof amounts to a computation. Let $\vec{u} = \ind_{N}$ and  $\vec{v}=\begin{pmatrix}
 a\ind_{\abs{I}}\\
 b\ind_{\abs{J}}\\
 \end{pmatrix}$, where $\ind_{\abs{J}}$ is a vector of all ones of length $\abs{J}$. Then $M\vec{u} =0$
 since the matrix $M$ has row sum 0. Another calculation shows that 
\begin{equation*}
L\vec{v}=(a-b)\begin{pmatrix}
 \beta\abs{J}\ind_{\abs{I}}\\
 -\beta\abs{I}\ind_{\abs{J}}\\
 \end{pmatrix}
\end{equation*}
Thus on the two dimensional subspace of vectors of the form $\vec{v}=a\ind_{\abs{I}}\oplus b\ind_{\abs{J}}\\$ 
we obtain that the action of $M\vec{v}$ is equivalent to 
\begin{equation*}
\begin{pmatrix}
\abs{J}\beta&-\abs{J}\beta\\
-\abs{I}\beta&\abs{I}\beta
\end{pmatrix}
\begin{pmatrix}
a\\
b
\end{pmatrix}=
(a-b)\begin{pmatrix}
 \beta\abs{J}\\
 -\beta\abs{I}\\
 \end{pmatrix}
\end{equation*}
The eigenvalues of this $2\times2$ matrix are easily computed to be $0$ and $N\beta$ with corresponding (unnormalized) eigenvectors $\ind_{N}$ and $\begin{pmatrix}
 \abs{J}\ind_{\abs{I}}\\
 -\abs{I}\ind_{\abs{J}}\\
 \end{pmatrix}$ respectively corresponding to $a-b=0$ and $a-b=N$.

Let $\vec{w}$ be a vector of the form $\vec{w}=(q_1,\ldots,q_{\abs{I}},0,\ldots,0)^{t}$, that is $\vec{w}\in \R^{\abs{I}}\oplus\vec{0}$. A computation shows that
\begin{equation*}
 M\vec{w} =
(\alpha\abs{I}+\beta\abs{J})\begin{pmatrix}
 q_1\\
 \vdots\\
q_{\abs{I}}\\
 {0}\\
\vdots\\
0
 \end{pmatrix}-
\sum_{i=1}^{\abs{I}}q_i\begin{pmatrix}
 \alpha\\
 \vdots\\
\alpha\\
 \beta\\
\vdots\\
\beta
\end{pmatrix}
\end{equation*}
Choosing $\vec{q}\in \R^{\abs{I}}$ such that $\sum\limits_{i=1}^{\abs{I}}q_i=0$ we see
 that $(\alpha\abs{I}+\beta\abs{J})$ is an eigenvalue with a $\abs{I}-1$ 
dimensional eigenspace since the equation $\sum\limits_{i=1}^{\abs{I}}q_i=0$ defines
a hyperplane through the origin in $\R^{\abs{I}}$.

Similarly, by considering the vector $\vec{p} =(0,\ldots,0,p_1,\ldots,p_{\abs{J}})^{t}\in\vec{0}\oplus\R^{\abs{J}}$,
 we get that $(\alpha\abs{J}+\beta\abs{I})$ is an eigenvalue with a $\abs{J}-1$ dimensional eigenspace.
\end{proof}

We can now define
\begin{defn}
 The bipartite state is one corresponding to the (unnormalized) vector $\begin{pmatrix}
 a\ind_{\abs{I}}\\
 b\ind_{\abs{J}}\\
 \end{pmatrix}$
where $I$ and $J$ are non-empty partitions of the vertex set $V(\Gamma)$.
\end{defn}

Note the obvious fact that the trivial partition corresponds to exactly to the consensus state. We now have the following:
\begin{lem}
For a complete graph there exist a critical point of the constrained gradient flow which is of bipartite form.
\end{lem}
\begin{proof}
 The critical points of the flow satisfy $\PD{\mathscr{D}}{x_i}=0$ and $\PD{\mathscr{D}}{\gamma_{ij}}=0$. Thus is follows that we must solve:
\begin{equation*}
\left.\begin{aligned}
 &2\sum_{i< j}\gamma_{ij}(x_i-x_j) -2\lambda x_i =0;&&(i=1\ldots N),\\
 &(x_i-x_j)^2-\frac{\mu}{\abs{E}}-\frac{2\tau}{\abs{E}}\gamma_{ij} =0;&&(i,j=1\ldots N).
\end{aligned}
\right\}
\end{equation*}
As noted earlier, this set of equations is equivalent to:
\begin{equation}\label{eqn:gamma_crit_pt}
 \gamma_{ij}=\frac{\abs{E}}{2\tau}(x_i-x_j)^2-\frac{\mu}{2\tau},
\end{equation}
and the eigenvalue equation 
\begin{equation}\label{eqn:x_crit_pt}
 {\bf L}(\gamma)\vx =\lambda\vx,
\end{equation}
 with the matrix ${\bf L}(\gamma)$ given by:
\begin{equation*}
 {\bf L}=
\begin{pmatrix}
\sum\limits_{j\neq 1}\gamma_{1j}&-\gamma_{12}&\ldots&-\gamma_{1N}\\
-\gamma_{21}&\sum\limits_{j\neq 2}\gamma_{2j}&\ldots&-\gamma_{2N}\\
\vdots&\vdots&\ddots&\vdots\\
-\gamma_{N1}&-\gamma_{N2}&\ldots&\sum\limits_{j\neq N}\gamma_{Nj}
\end{pmatrix}.
\end{equation*}
 We will now show that we have a solution of the form \eqref{eqn:bipartite}. We choose the
 eigenvector $\vx$ corresponding to the eigenvalue $\lambda=N\beta$ normalized by the contraint $\norm{\vx}^2 = R$.
In other words we choose 
\[\vx=\sqrt{\frac{R}{\abs{J}\abs{I}N}}\begin{pmatrix}
 \abs{J}\ind_{\abs{I}}\\
 -\abs{I}\ind_{\abs{J}}\\
 \end{pmatrix}.\]
For $i\neq j$ in the same subsets we have that $x_i-x_j =0$ and thus, using \eqref{eqn:gamma_crit_pt}, we see that 
\[\alpha=-\frac{\mu}{2\tau}.\]
For $i,j$ in different subsets, we see that $(x_i-x_j)^2 =\dfrac{RN}{\abs{I}\abs{J}}$ and using \eqref{eqn:gamma_crit_pt} we get that
\[\beta =\frac{\abs{E}}{2\tau}\frac{RN}{\abs{I}\abs{J}}-\frac{\mu}{2\tau}.\]
Since the graph is complete, it follows that there are $\abs{I}(\abs{I}-1)/{2}$ and $\abs{J}(\abs{J}-1)/{2}$ edges in the ``cliques'' with $\abs{I}$ and $\abs{J}$
vertices respectively, and $\abs{I}\abs{J}$ edges between the two cliques. Thus the constraints \eqref{eqn:gam_sum_cons} and \eqref{eqn:gam_norm_cons} then imply
\begin{align*}
 Q&=\frac{2}{N(N-1)}\brac{\frac{\abs{I}(\abs{I}-1)}{2}+\frac{\abs{J}(\abs{J}-1)}{2}}\alpha +\frac{2\abs{I}\abs{J}\beta}{N(N-1)}\\
 P&=\frac{2}{N(N-1)}\brac{\frac{\abs{I}(\abs{I}-1)}{2}+\frac{\abs{J}(\abs{J}-1)}{2}}\alpha^2 +\frac{2\abs{I}\abs{J}\beta^2}{N(N-1)}\\
\end{align*}
Solving this system of equations is straightforward but lengthy. When the dust settles, we obtain
\begin{align}
\alpha &= Q\pm\nu\sqrt{\frac{r}{1-r}}\\
\beta &= Q\mp\nu\sqrt{\frac{1-r}{r}}
\end{align}
where 
\[
\nu^2 = P-Q^2,
\] 
and
 \[
 r=\dfrac{\abs{I}\abs{J}}{\abs{E}}
 \]
 is the fraction of the total number of edges that connect the different cliques. Now
\[\frac{R\abs{E}N}{2\tau\abs{I}\abs{J}} =\beta-\alpha = \mp\nu\brac{\sqrt{\frac{1-r}{r}}+\sqrt{\frac{r}{1-r}}}\]
Hence
\begin{equation}
\tau =\mp\frac{RN}{2\nu rk(r)},
\end{equation}
and
\begin{equation}
\mu =\pm\frac{RN}{\nu rk(r)}\paren{Q\pm\nu\sqrt{\frac{r}{1-r}}},
\end{equation}
while 
\begin{equation}
\lambda:=N\beta=N\paren{Q\mp\nu\sqrt{\frac{r}{1-r}}}
\end{equation}
where we have defined \[k(r)=\sqrt{\frac{1-r}{r}}+\sqrt{\frac{r}{1-r}} =\frac{1}{\sqrt{r(1-r)}}.\]
This completes the proof.
\end{proof} 

\section{Stability of the Consensus State}
\subsection{Global Stability}
 To motivate the results, we assume first that $\epsilon =0$ so that the gradient flow becomes
\begin{equation}
\left.\begin{aligned}
 \dot{\vx}&=-2(\vl\vx-\lambda\vx)\\
 \dot{\gamma}_{ij}&=0.
\end{aligned}
\right\}
\end{equation}
Recall that the constraints imply $\vx\cdot\dot{\vx} =0$ from which we derived
\begin{equation*}
\lambda = \frac{\dpair{\vl\vy}{\vy}}{\norm{\vx}^2}.
\end{equation*}
Writing $\vx(t) =u(t)\ind_{N}+\vy(t)$ with $\vy\cdot\ind_N =0$ we see that
\begin{align*}
\dot{u}&= 2 u\frac{\dpair{\vl\vy}{\vy}}{\norm{\vy}^2 + Nu^2}\\
\dot{\vy}&= -2\paren{\vl\vy -\vy\frac{\dpair{\vl\vy}{\vy}}{\norm{\vy}^2 + Nu^2}}.
\end{align*}
Let $v(t)=\norm{\vy(t)}$. It follows from the constraint that \[ Nu^2(t)+v^2(t) = Nu^2(0)+v^2(0)= R.\]
In particular
\[
v\dot{v}=\frac{1}{2}\DO{t}v^2 =-2\paren{\dpair{\vl\vy}{\vy} -\dpair{\vy}{\vy}\frac{\dpair{\vl\vy}{\vy}}{R}},
\]
which implies that
\begin{equation*}
\dot{v} = -2v\paren{1-\frac{v^2}{R}}\frac{\dpair{\vl\vy}{\vy}}{\norm{\vy}^2}.
\end{equation*}
Suppose first that $v(0) =0$. Then one sees that $u(t) = \sqrt{\frac{R}{N}}$, that is $v(t) =0$ for all $t$.
Now assume that $v(0) \in (0,\sqrt{R}]$ and that the matrix $\vl$ has kernel precisely $\ind_N$ with
all other eigenvalues positive. It then follows that
\[
\frac{\dpair{\vl\vy}{\vy}}{\norm{\vy}^2}\geq \sigma_{min} >0,
\]
from which we see that $v$ is monotone decreasing and in particular that 
\begin{equation*}
\dot{v} \leq -2\sigma_{min}v\paren{1-\frac{v^2(0)}{R}}.
\end{equation*}
A direct argument or the use of Gronwall's inequality shows that $v=0$ is exponentially attracting.
Thus we have proved

\begin{thm}[Stability of Consensus State I]\label{thm:stability_cons_I}
Suppose $\epsilon=0$. Suppose also that ${\vl}(0)$ is positive semi-definite with a 1 dimensional kernel.
Then it holds that 
\begin{equation*}
\lim_{t\to\infty}\vx(t) =\sqrt{\frac{R}{N}}\ind_N
\end{equation*}
That is, the consensus state is globally asymptotically stable.
\end{thm}

The demonstration above relied mostly on the fact that we had a \emph{spectral gap}. 
We know that $0$ is always an eigenvalue of $\vl$  but the assumption $\epsilon =0$ was enough
to guarantee that the next eigenvalue was strictly positive if $\vl(0)$ had all nonnegative eigenvalues.
If $\epsilon\neq 0$, we can still find sufficient conditions guaranteeing the existence of a spectral gap.

\begin{thm}[Sufficient Conditions for Global Stability]
Suppose that 
\[
\frac{P-Q^2}{Q^2}<\frac{1}{N-1}.
\]
Then $\vl$ is always positive semi-definite and the consesnsus state is a global minimizer.
\end{thm}
\begin{proof}
For the proof we will verify that, under the stated assumption, on $P,Q$, $\vl$ is positive semi-definite with
a $1$-dimensional kernel. We write each opinion as a mean plus a mean-zero part,
\[
\gamma_{ij} = Q + \tilde \gamma_{ij}
\]
where mean-zero part $\tilde\gamma_{ij}$ now satisfies 
\begin{align}
\frac{1}{E}\sum_{i<j}\tilde \gamma_{ij}&=0 \\
\frac{1}{E}\sum_{i<j}\tilde \gamma_{ij}^2 &=P-Q^2. 
\end{align}
The corresponding graph Laplacian takes the form 
\begin{align}
{\vl} & = {\vl}_0 + \tilde {\vl} \\
& = Q\left(\begin{array}{cccc}(N-1)  & -1 & -1 & \ldots \\
-1 & (N-1)  & - 1 & \ldots \\ 
-1 & -1 & (N-1)  & \ldots \\
\vdots & \vdots & \vdots & \ddots 
\end{array}\right) \\ \nonumber &+ \left(\begin{array}{cccc}\sum\limits_{i\neq 1}\tilde\gamma_{i1}& -\tilde\gamma_{12} & -\tilde\gamma_{13}& \ldots \\ -\tilde\gamma_{12} & \sum\limits_{i\neq 2}\tilde\gamma_{12} & -\tilde\gamma_{23} & \ldots \\
-\tilde\gamma_{13} & -\tilde\gamma_{23} & \sum\limits_{i\neq 3}\tilde \gamma_{i3} & \ldots \\
 \vdots & \vdots & \vdots & \ddots 
\end{array}
\right)
\end{align}
The important observation is that the matrix ${\vl}_0$ commutes with 
every graph Laplacian, and thus they can be simultaneously diagonalized, 
and we need only estimate the  most negative eigenvalue of $\tilde {\vl}$. 
The latter is easily estimated in terms of the Hilbert-Schmidt inequality. We 
have 
\[
\sigma_{min}(\tilde {\vl}) \geq - \paren{\sum_i \paren{\sigma_i(\tilde {\vl})}^2}^{\frac12}=-\norm{\tilde {\vl}}_{HS} 
\]
and 
\begin{align*}
\norm{\tilde {\vl}}_{HS}^2 &= \sum_{i,j}\tilde \gamma_{ij}^2 + \sum_i \paren{\sum_{j\neq i}\tilde \gamma_{ij}}^2 \\
&\leq  \sum_{i,j}\tilde \gamma_{ij}^2 + \sum_i \paren{\sum_{j\neq i}\tilde \gamma_{ij}^2}(N-1) \\
&\leq 2 E (P-Q^2) + 2(N-1)E(P-Q^2)\\
& \leq 2 N E (P-Q^2) 
\end{align*}
This gives the inequality for ${\vl}$ that the minimum eigenvalue, {\em other than the zero eigenvalue of course}, satisfies the estimate
\[
\sigma_{min}({\vl}) \geq N\paren{Q - \sqrt{(N-1) (P-Q^2)}} >0.
\]
Thus once again we have a spectral gap and the proof of Theorem \ref{thm:stability_cons_I} can be repeated almost verbatim.
\end{proof}

There is a nice geometric interpretation of this result. As we have already observed, the constraint $g_1=Q$ defines a hyperplane in $\R^{\abs{E}}$ while
the constraint $g_2=P$ defines a sphere of radius $\sqrt{P}$ in $\R^{\abs{E}}$. If $P=Q^2$, then they intersect tangentially at the
single point $\gamma^{*}_{ij}=Q$, i.e. $\vg^{*}=Q(1,1,\ldots,1)^t$. Note that $\vl(\gamma^{*}_{ij})=\vl_0$ and that the spectrum of $\vl_0$ consists of
$0$ which is simple and $NQ$ of multiplicity $N-1$. This is clearly positive semi-definite with a one dimensional kernel as long as $Q>0$. 
What we have really shown is that for $P-Q^2\sim$ small then $\vl$ still satisfies this property as well.

\begin{thm}
Let $\epsilon>0$ be arbitrary but fixed. Suppose that $\sigma_{min}(\vm(0)) >0$ i.e., $\vl(0)$ is positive semi-definite with a 1 dimensional kernel.
Then there is a positive $\delta = \delta(\epsilon,\sigma_{min}(\vm(0)),R,N,P,Q)$ such that $\theta=0$ is attracting in the neighborhood $N_{\delta}(0)$.
In other words, if $\vl(0)$ satisfies the stated hypothesis then the consensus state is locally stable.
\end{thm}

\section{Stability of The Bipartite State}
In this section we shall address the issue of stability of the explicitly constructed bipartite states. The approach is to linearize the flow about the bipartite equilibria and to count the number of negative eigenvalues, i.e. the index, of the resulting linear map. 

Recall that the critical points of the flow are precisely the constrained extrema of the Dirichlet energy $D$ i.e. the \emph{critical points}
 of $\mathscr{D}$. The method of Lagrange multipliers and a standard Lyapunov function argument gives the following result
\begin{lem}
Let $\omega_0\in\Omega$ be a critical point of $\mathscr{D}$ i.e. a local extrema of $D$ subject to the constraints \eqref{eqn:gam_sum_cons}, \eqref{eqn:gam_norm_cons}, \eqref{eqn:x_norm_cons}. Let the associated Hessian 
\begin{equation}
H(\omega_0)=-
\at{\begin{pmatrix}
\displaystyle \PDDM{\mathscr{D}}{x_i}{x_k}&\displaystyle \PDDM{\mathscr{D}}{\gamma_{lk}}{x_i}\\
\\
\epsilon\displaystyle \PDDM{\mathscr{D}}{x_k}{\gamma_{ij}}&\epsilon\displaystyle \PDDM{\mathscr{D}}{\gamma_{lk}}{\gamma_{ij}}
\end{pmatrix}}{\omega_0}.
\end{equation}
Let $T_{\omega_0}\Omega$ be the tangent space to $\Omega$ at $\omega_0$. If $\at{H(\omega_0)}{T_{\omega_0}\Omega}$ is negative definite, then the gradient flow is stable near $\omega_0$.
\end{lem}

Thus we need only determine the number of negative eigenvalues of the Hessian, $H$, restricted to $T\Omega$. To facilitate the computation we need the following

\begin{lem}
Let $A$ be a symmetric invertible matrix in $\R^n$. Given any linear subspace $S\subset\R^n$ then \[n_{\pm}(A) = n_{\pm}(\at{A}{S})+ n_{\pm}(\at{A^{-1}}{S^{\perp}}).\]
\end{lem}

\subsection{Index of the Full Hessian}
We first put coordinates on $\R^{|E|}$ by ordering the pairs $(i,j)$ with $i<j$ lexicographically. That is we write
\[\vg = (\gamma_{12},\gamma_{13},\ldots,\gamma_{23},\gamma_{24},\ldots,\gamma_{34},\ldots).\]
A direct computation shows that
\begin{align*}
\partial^2_{\vx}{\mathscr{D}} = 2(\vl(\vg)-\lambda)\qquad
\text{and}\qquad
 \partial^2_{\vg}{\mathscr{D}} =-\frac{2\tau}{|E|}I_{|E|\times|E|}.
 \end{align*}
 Furthermore
 \begin{align*}
  \PDDM{\mathscr{D}}{\gamma_{lk}}{x_i}=2(x_i-x_j)\delta_{ij,lk},\qquad\text{and}\qquad
 \PDDM{\mathscr{D}}{x_k}{\gamma_{ij}}=2(x_i-x_j)(\delta_{ik}-\delta_{jk}).
 \end{align*}
Recall that the bipartite equilibrium is given by 
\[
\vx=\sqrt{\frac{R}{\abs{J}\abs{I}N}}\begin{pmatrix}
 \abs{J}\ind_{\abs{I}}\\
 -\abs{I}\ind_{\abs{J}}
 \end{pmatrix}.
 \] 
 Defining the $N\times|E|$ matrix $B$ with entries

 \begin{equation}
 B_{i,lk}=2\sqrt{\frac{RN}{\abs{J}\abs{I}}}
 \begin{cases}
 1&i\in\{l,k\};\quad i\in I;\quad\{l,k\}\setminus\{i\}\in J\\
 -1&i\in\{l,k\};\quad i\in J;\quad\{l,k\}\setminus\{i\}\in I,
 \end{cases}
 \end{equation}
 we see that $\PDDM{\mathscr{D}}{\gamma_{lk}}{x_i}=B_{i,lk}$. To summarize, we have shown that
  \begin{equation*}
 H=
 \begin{pmatrix}
 -2(\vl(\vg)-\lambda)&-B\\
 -\epsilon B^t&\dfrac{2\epsilon\tau}{|E|}I_{|E|\times|E|}
 \end{pmatrix}.
 \end{equation*}
 
 The next result will prove useful in simplifying the computations:
 \begin{lem}[Schur Formula]
 Suppose that $M$ is a symmetric matrix of the form
 \begin{equation*}
 M=
 \begin{pmatrix}
 A&B\\
 B^t&C
 \end{pmatrix}
 \end{equation*}
 where $A$, $B$ and $C$ are $m\times m$, $m\times k$ and $k\times k$ matrices respectively and $C$ is invertible. Then
 \[n_{\pm}(M) = n_{\pm}(C) +n_{\pm}(A-BC^{-1}B^t).\]
 \end{lem}
 
As a consequence we have
\[n_{-}(H) = n_{-}(\frac{2\epsilon\tau}{|E|}I_{|E|\times|E|})+n_{-}( -2(\vl(\vg)-\lambda)-\frac{|E|}{2\tau}BB^t).\] Since 
\begin{equation*}
n_{-}\paren{\frac{2\epsilon\tau}{|E|}I_{|E|\times|E|}} =
\begin{cases}
 0& \quad \tau\geq 0\\
 |E|& \quad \tau< 0
 \end{cases}
 \end{equation*}
  it follows that
 \begin{prop}
 $\tau<0$ is a necessary condition for the bipartite state to be stable.
 \end{prop}
 
 Observe that 
 \begin{equation*}
 BB^t=\frac{4RN}{|I||J|}
 \begin{pmatrix}
       |J| &         0  &  \ldots &         0 &        -1 &       -1 & \ldots & -1\\
         0 &        |J| &  \ldots &         0 &       - 1 &       -1 & \ldots & -1\\
 \vdots & \vdots & \ddots &\vdots & \vdots & \vdots & \vdots & \vdots\\
          0 &         0 &   \ldots&       |J| &        -1 &        -1 &  \ldots & -1\\
 -1&-1&\ldots&-1&|I|&0&\ldots&0\\
 -1&-1&\ldots&-1&0&|I|&\ldots&0\\
\vdots&\vdots&\vdots&\vdots&\vdots&\vdots&\ddots&\vdots\\
-1&-1&\ldots&-1& 0&0&\ldots&|I|
 \end{pmatrix}
 \end{equation*}
 and recall that the eigenbasis of $\vl(\vg)$ at the critical point is spanned by the vectors
  \begin{align*}\left\{
  \ind_N,\quad
  \begin{pmatrix}
 \abs{J}\ind_{\abs{I}}\\
 -\abs{I}\ind_{\abs{J}}\\
 \end{pmatrix},\quad
 \begin{pmatrix}
\vec{q}\\
0\\
 \end{pmatrix},\quad
 \begin{pmatrix}
0\\
\vec{p}\\
 \end{pmatrix}:\quad\sum_{i=1}^{|I|}q_i=\sum_{j=1}^{|J|}p_j=0\right\}.
  \end{align*}
Now, a direct computation shows
\begin{align*}
BB^t  \ind_N &=0\\
BB^t  \begin{pmatrix}
 \abs{J}\ind_{\abs{I}}\\
 -\abs{I}\ind_{\abs{J}}\\
 \end{pmatrix} & =\frac{4RN^2}{|I||J|}\begin{pmatrix}
 \abs{J}\ind_{\abs{I}}\\
 -\abs{I}\ind_{\abs{J}}\\
 \end{pmatrix} \\
 BB^t \begin{pmatrix}
\vec{q}\\
0\\
 \end{pmatrix}&=\frac{4RN}{|I|}\begin{pmatrix}
\vec{q}\\
0\\
 \end{pmatrix} \\
 BB^t \begin{pmatrix}
0\\
\vec{p}\\
 \end{pmatrix}&=\frac{4RN}{|J|} \begin{pmatrix}
0\\
\vec{p}\\
 \end{pmatrix};
\end{align*}
meaning we have verified 
   \begin{prop}
 The matrix $BB^t$ and $\vl(\vg)$ have the same eigen-basis. In particular the spectrum of $BB^t$ is given by
  \begin{equation}
  \sigma(BB^t)=
 \begin{cases} 
 0,&\text{\,\, of multiplicity}=1\\
\dfrac{4RN^2}{|I||J|},&\text{\,\, of multiplicity}=1\\
 \dfrac{4RN}{|I|},&\text{\,\, of multiplicity}=\abs{I}-1\\
  \dfrac{4RN}{|J|},&\text{\,\, of multiplicity}=\abs{J}-1.\\
 \end{cases}
 \end{equation}
 \end{prop}
 
 Since $\lambda=N\beta$ for the bipartite state, it follows that the spectrum of $(\vl(\vg_0)-\lambda)+\frac{|E|}{4\tau}BB^t$ is contained in the set
 \[\left\{ -N\beta,\quad \frac{N^2R}{r\tau},\quad\frac{RN}{2r\tau}(2\abs{J}-\abs{I}),\quad\frac{RN}{2r\tau}(2\abs{I}-\abs{J}) \right\},\]
 the latter two being repeated $\abs{I}-1$ and $\abs{J}-1$ times respectively. For $\abs{I}> 2N/3$, we see that $2\abs{J}-\abs{I}<0$ and for 
 $\abs{I}<N/3$, we see that $2\abs{I}-\abs{J}<0$. Putting all this information together we have proven
 
 \begin{prop}
 Assume $\tau<0$. For convenience put \[T=-2(\vl(\vg)-\lambda)-\frac{|E|}{2\tau}BB^t.\] 
 Then
\begin{equation}
 \left.\begin{aligned}
 n_{-}(T )=
 \begin{cases}
 \abs{J}-1,&\abs{I}\in[1,N/3)\\
0,&\abs{I}\in (N/3,2N/3)\\
\abs{I}-1,&\abs{I}\in(2N/3,N-1]\\
 \end{cases}
\end{aligned}\right\}\quad\text{for\,\,}\beta>0;
 \end{equation}
 and
  \begin{equation}
 n_{-}(T)=
  \left.\begin{aligned}
 \begin{cases}
 \abs{J},&\abs{I}\in[1,N/3)\\
1,&\abs{I}\in (N/3,2N/3)\\
\abs{I},&\abs{I}\in(2N/3,N-1]\\
 \end{cases}
\end{aligned}\right\}\quad\text{for\,\,}\beta<0.
 \end{equation}
 \end{prop}
 
 As a consequence we have the following
 \begin{corr}
 Assume $\tau<0$. If $\abs{I}=1$ or $\abs{I}=N-1$ then $n_{-}(H)= \abs{E}+N-2$ for $\beta>0$ and  $n_{-}(H)= \abs{E}+N-1$ for $\beta<0$.
 \end{corr}

 \subsection{Index of the ``Reduced Hessian"}
 Our goal now is to compute the index: $n_{-}(\at{H}{(T_{\omega_0}(\Omega))^\perp})$ of the Hessian restricted to the orthogonal complement of $T_{\omega_0}(\Omega)$.  We begin with

\begin{lem}
Let $\omega_0$ be the bipartite critical point and put $S=T_{\omega_0}(\Omega)$. If $H$ is invertible then
\[n_{-}(\at{H^{-1}}{S^\perp}) _{+}(\nabla_{(\mu,\tau,\lambda)}(g_1,g_2,g_3)^t).\]
\end{lem}

\begin{proof}
 This is really a fact from the method of Lagrange multipliers amd the thery of constrained optimization in general. First note that
 \[\nabla_{(\vx,\vg)}D-\mu\nabla_{(\vx,\vg)}g_1-\tau\nabla_{(\vx,\vg)}g_2-\lambda\nabla_{(\vx,\vg)}g_3 =0,\]
and in general we can determine $(\vx,\vg)=(\vx(\tau,\mu,\lambda),\vg(\tau,\mu,\lambda))$ as functions of the Lagrange multipliers. Differentiating
the above expression with respect to, say, $\mu$ by the usual chain rule gives
\[-H\cdot\PD{\vx}{\mu}=\nabla_{(\vx,\vg)}g_1,\]
and since $H$ is invertible we see that 
\[-\PD{\vx}{\mu}=H^{-1}\cdot\nabla_{(\vx,\vg)}g_1.\]
Similarly 
\[-\PD{\vx}{\tau}=H^{-1}\cdot\nabla_{(\vx,\vg)}g_2,\qquad\text{and}\qquad -\PD{\vx}{\lambda}=H^{-1}\cdot\nabla_{(\vx,\vg)}g_3.\]
Since $S^\perp$ is spanned by $\{\nabla_{(\vx,\vg)}g_i\}_{i=1}^3$, any $\vec{v}\in S^\perp$ can be written as $\vec{v} = \sum_i\alpha_idg_i$.
Consequently
\begin{align*}
 \at{H^{-1}}{S^\perp\times S^\perp}&=(H^{-1}\vec{v},\vec{v})\\
  &= \dpair{H^{-1}\paren{\sum_{i=1}^3\alpha_i\nabla_{(\vx,\vg)}g_i}}{\sum_{j=1}^3\alpha_j\nabla_{(\vx,\vg)}g_j}\\
&=-\sum_{i,j}\alpha_i\alpha_j\dpair{\PD{\vx}{\mu_i}}{\nabla_{(\vx,\vg)}g_j}\\
&=-\sum_{i,j}\alpha_i\alpha_j\PD{g_j}{\mu_i}\\
&=-\at{\nabla_{(\mu,\tau,\lambda)}(g_1,g_2,g_3)^t}{\R^3\times\R^3},
\end{align*}
whence the result.
\end{proof}

 Next recall that we have the following set of equations
 
 \begin{align*}
 \mu+2\tau M_1&=NR\\
 \mu M_1+2\tau M_2&=R\lambda\\
 \frac{\lambda}{N}+\frac{\mu}{2\tau}&=\frac{RN}{2r\tau}
 \end{align*}
 Solving for $R$, $P$ and $Q$ and using the fact that $g_1=Q$, $g_2=P$ and $g_3=R$, we get
 \begin{align}
 &g_1(\mu,\tau,\lambda)=\frac{r\lambda}{N}-\frac{\mu(1-r)}{2\tau}\\
& g_2(\mu,\tau,\lambda)=\frac{r\lambda ^2}{N^2}+\frac{\mu^2(1-r)}{4\tau^2}\\
 &g_3(\mu,\tau,\lambda)=\frac{2r\tau}{N}\paren{\frac{\lambda}{N}+\frac{\mu}{2\tau}}.
 \end{align}
 
 The Jacobian is
 \begin{equation}
 \begin{pmatrix}
\displaystyle \PD{g_3}{\lambda}& \displaystyle \PD{g_3}{\mu}& \displaystyle \PD{g_3}{\tau}\\
\\
 \displaystyle\PD{g_1}{\lambda}& \displaystyle \PD{g_1}{\mu}&\displaystyle \PD{g_1}{\tau}\\
 \\
  \displaystyle\PD{g_2}{\lambda}&\displaystyle \PD{g_2}{\mu}&\displaystyle \PD{g_2}{\tau}
 \end{pmatrix}
 =
 \begin{pmatrix}
\displaystyle \frac{2r\tau}{N^2} & \displaystyle \frac{r}{N} & \displaystyle\frac{2r\lambda}{N^2}\\
\\
\displaystyle \frac{r}{N}&-\displaystyle\frac{(1-r)}{2\tau}& \displaystyle \frac{\mu(1-r)}{2\tau^2}\\
\\
\displaystyle \frac{2r\lambda}{N^2}& \displaystyle \frac{\mu(1-r)}{2\tau^2}& \displaystyle\frac{\mu^2(1-r)}{2\tau^3}
 \end{pmatrix}
 \end{equation}
and  computing the determinant of the principal minors we get 
 \[\Delta_1= \frac{2r\tau}{N^2};\quad\Delta_2= -\frac{r}{N};\quad \text{and}\quad\Delta_3= \frac{2r^2(1-r)}{N^2\tau}\paren{\frac{\mu}{2\tau}+\frac{\lambda}{N}}^2.\]
 Thus for $\tau<0$, we see that $n_{-}(\at{H^{-1}}{S^\perp})=2$ and if in addition $\beta<0$, we get
\[n_{-}(\at{H(\omega_0)}{T_{\omega_0}\Omega})= n_{-}(H(\omega_0)) -n_{-}(\at{H(\omega_0)}{(T_{\omega_0}(\Omega))^\perp}) =N+\abs{E}-3\]
 which equals $\dim(T_{\omega_0}\Omega)$. As such, we have established

\begin{thm}
 Suppose that $\tau$, $\beta$ are both negative corresponding to the bipartite state where $\abs{I}=1$ or $\abs{I}=N-1.$\footnote{It might be more apt to call this
 particular state an ``ostracized state'' but perhaps this has a strong negative connotation.} Then this critical point is a local minimum of the constrained 
Dirichlet energy and as such is locally stable. If $\tau<0$ and $\beta>0$ then this state has a $1$-dimensional unstable manifold.
\end{thm}

\section{Numerical Simulations}

All simulations are conducted on the complete graph $K_5$, so all actors are 
known to one another ($N=5,E=10$). The relevant ODEs were integrated with a 
fourth order Runge-Kutta algorithm with $dt=10^{-3}$.  
In each case the initial positions $x_i(0)$ and the initial opinions $\gamma_{ij}(0)$ were chosen randomly.  The edges are chosen as follows: 
we generate a vector ${\bf \gamma}$ with independent, identically distributed 
entries drawn from the uniform distribution on $(-1,1)$. We then removed 
the mean of  ${\bf \gamma}$, and  scaled ${\bf \gamma}$  to have unit length. 
The edge weights were taken to be $(P-Q^2)^{\frac12}  {\bf \gamma} + Q(1,1,1,\ldots,1)^t$, so that the resulting edge weights have mean $Q$ and variance $P-Q^2$. The positions were chosen uniformly from $(-1,1)$ with no mean.

The first set of graphs show the evolution of the positions (vertex weights) 
and opinions (edge weights) for an initial condition that converges to a 
consensus state. Note that initially two opinions $\gamma_{ij}$ are 
negative, and system converges to a stable consensus state with one negative 
opinion. This illustrates that the model can accomodate a certain amount of 
imbalance if a consensus is reached: actors can overcome a some antipathy if 
they have a common cause.

\begin{figure}[ht]
 \includegraphics[width=.3\textwidth]{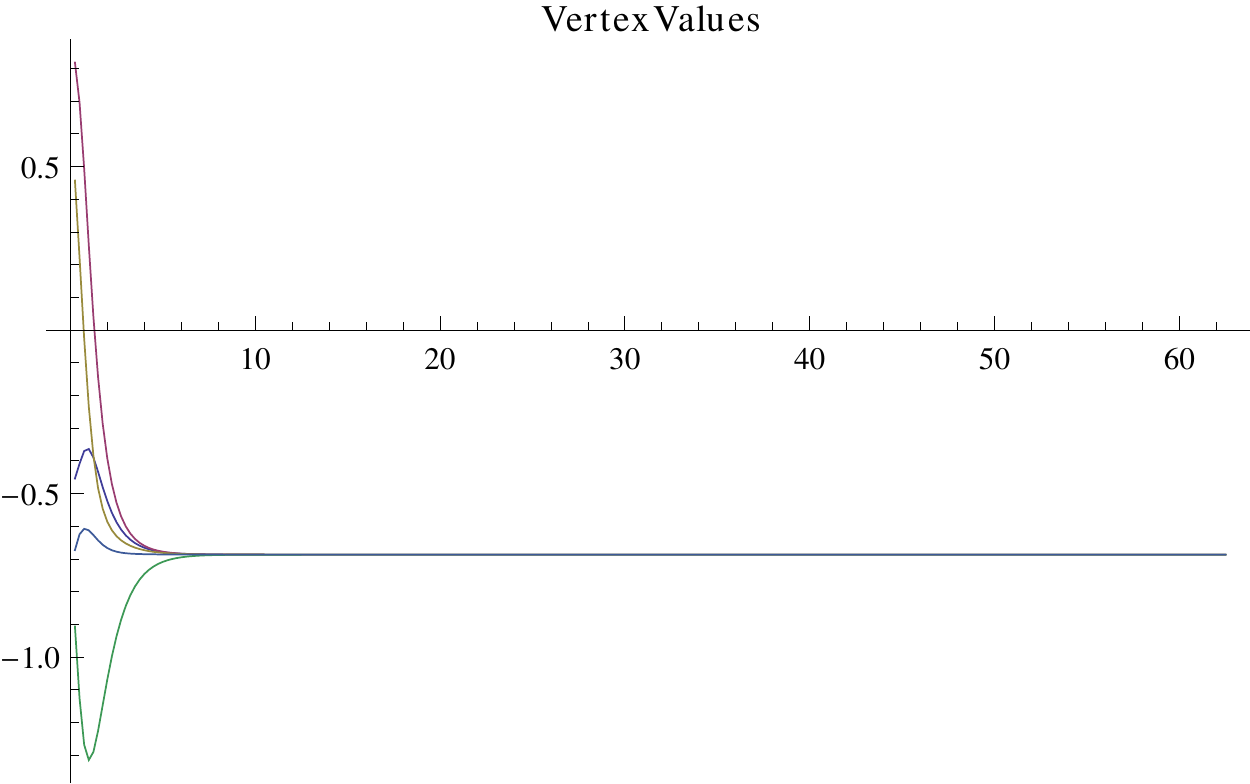}\\
\includegraphics[width=.3\textwidth]{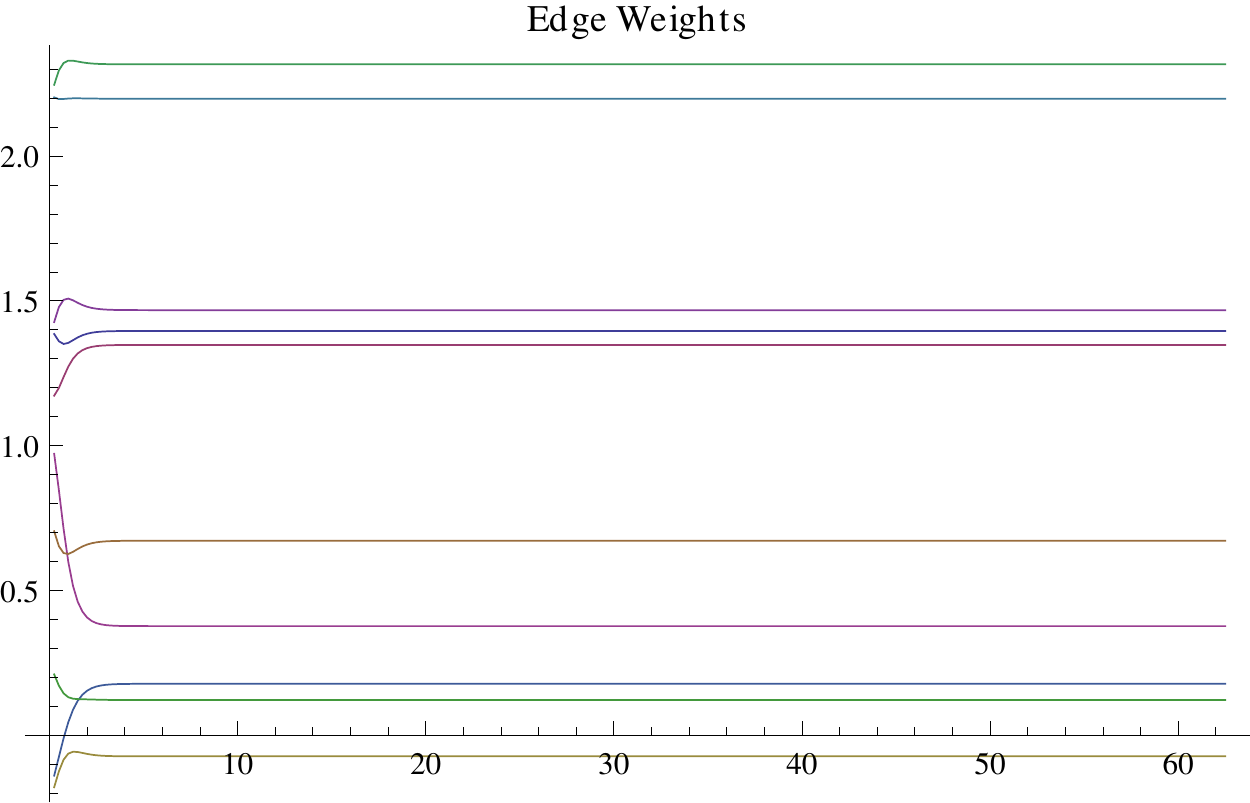}
\caption{The evolution of the positions (top) and opinions
(bottom) for initial conditions 
that converge to a consensus.}
\end{figure}

The second numerical experiment shows the dynamics in a case where the 
initial variance in the opinions, $P-Q^2$, is larger. In this case the 
dynamics does not converge to a consensus but rather to a balanced 
non-consensus state. In this case the actors divide into two parties, one 
with $4$ individuals and one with a single individual, 
where each actor has a positive opinion of the actors in the same camp and 
a negative opinion of the actors in the other camp.

\begin{figure}[ht]
  \includegraphics[width=.3\textwidth]{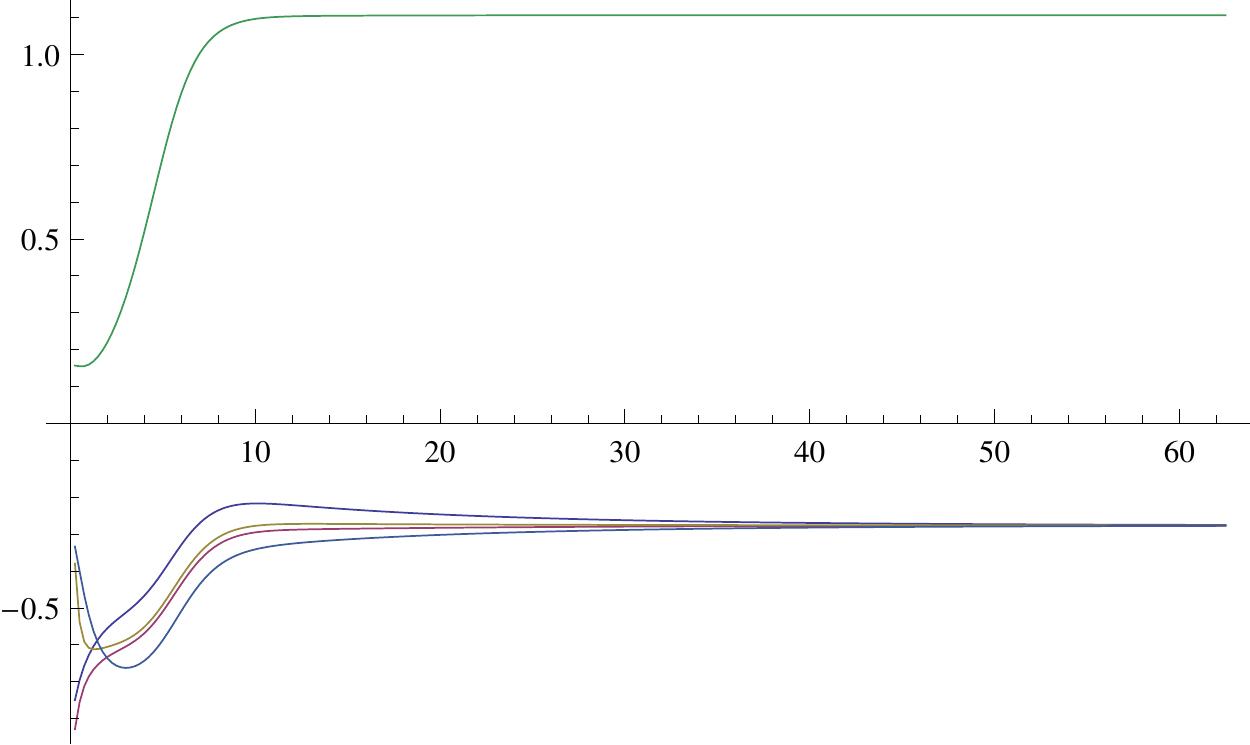}
\\ \includegraphics[width=.3\textwidth]{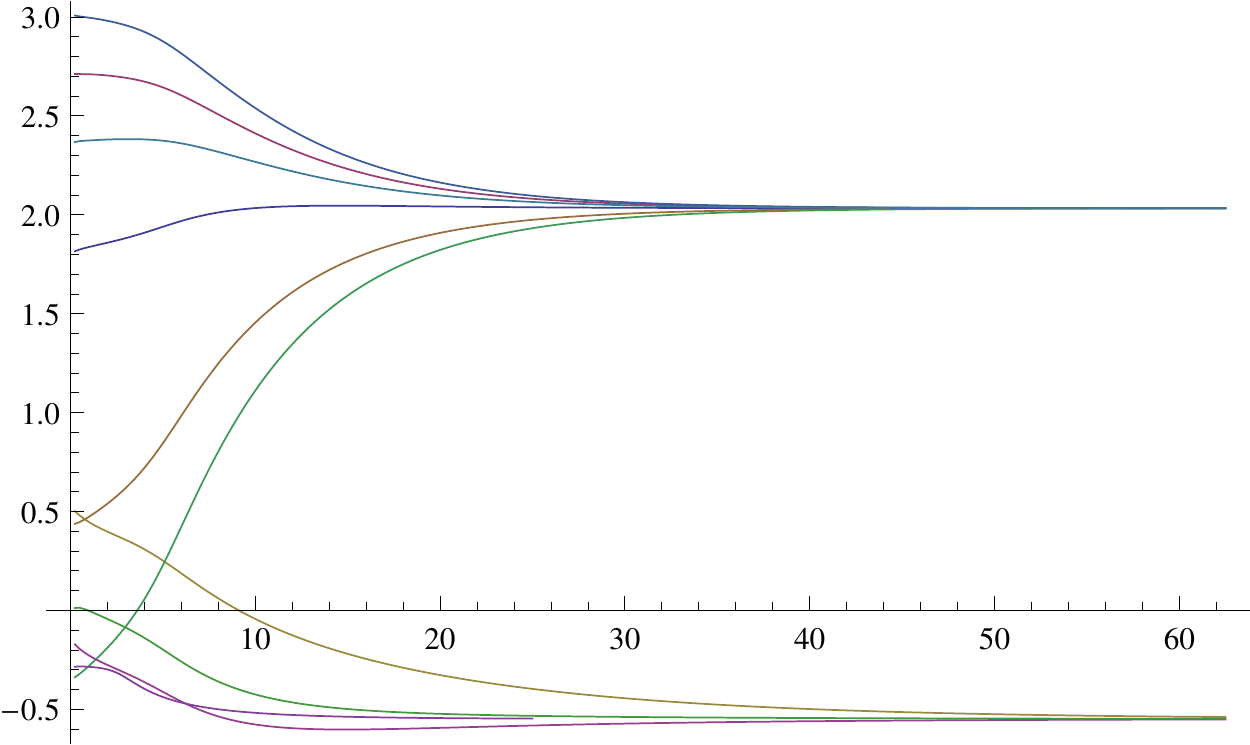}
 \caption{The evolution of the opinions (top) and positions 
(bottom) where the system evolves to a non-consensus balanced state.}
  \label{fig:evolve1}
\end{figure}

We have also considered a related system where there are contraints placed on
each individual actor, not just on the actors as a group. 

\section{ Conclusions}

We have introduced a model of dynamics on a network where the positions of 
individuals on some issue and the relationships between individuals co-evolve 
under a very natural dynamics. We see that the idea of balance arises 
naturally from a consideration of the steady states: all of the steady states
we have been able to find, with the exception of the consensus state, are 
either balanced states or anti-balanced states, with the latter always being 
unstable. In fact the only stable steady states that we have found analytically
or been able to observe numerically are the consensus state and the bi-partite 
state where one party has $1$ member and the other party has $N-1$ members.  

The latter fact seems somewhat surprising, and we believe that it is a 
consequence of the slightly unrealistic nature of the constraints on the 
opinions. In this letter we assume only global constraints on the opinions:
$E^{-1} \sum_{ij} \gamma_{ij}=Q, E^{-1} \sum_{ij} \gamma_{ij}^2=P,$ with no constraints
on individual actors. It would, perhaps, be more realistic to require that 
each actor maintain a certain mean level of civility: for each actor $i$ 
we could require that $(N-1)^{-1}\sum_{j}\gamma_{ij} = Q_i.$ Since the general 
effect of constraints is to increase the stability we expect that this 
type of constraint would lead to more stable steady states with parties of 
many different sizes. The analysis becomes more difficult in this case, however, and we leave this problem for future works.

It would also be interesting to consider more complicated graph topologies than 
the complete graph. Unlike the models of Antal, Krapivsky and 
Redner and  Kulakowski, Gawr\'onski and Gronek the model considered here 
extends naturally to an arbitrary graph which may contain few or no triangles. 
However it is not clear to what extent the balanced steady states for the 
complete graph persist in these more sparse graph topologies.


\bibliography{PRLBalance}

\section{Appendix}
In this section we show the existence of a locally attracting neighborhood of the consensus state.
The argument is similar in spirit to that for global stability. The main complication is that we have to
explicitly handle the evolution equation for the spectrum of $\vl$.

For convenience, we define $\hat{\ind}_N =\frac{1}{\sqrt{N}}\ind_N$ to be the normalized vector of all
 ones. Let  $W=\{\vy\in\R^N\,|\,\vy\cdot\hat{\ind}_N =0\}$ be the subspace of vectors orthogonal to the
 ``consensus state". Following a similar argument as before, we will write
\begin{equation*}
\vx = c\paren{\cos(\theta)\hat{\ind}_N +\sin(\theta)\vy}
\end{equation*}
with $\vy \in W$ and $\norm{\vy}=1$. The goal is to show that $\theta \to 0$ which will imply the stability
of the consensus state. We will first obtain equations governing the dynamics of the relevant variables.

From the constraint $\vx\cdot\dot{\vx} =0$, it follows that
\begin{equation*}
\lambda =\sin^2(\theta)\dpair{\vl\vy}{\vy}.
\end{equation*}
The original equation $\dot{\vx}=-2(\vl\vec{x}-\lambda\vx)$ and the fact that $\vl\hat{\ind}_N =0$ then implies
\begin{equation*}
\paren{-\sin(\theta)\dot{\theta}\hat{\ind}_N+\cos(\theta)\dot{\theta}\vy+\sin(\theta)\dot{\vy}}=-2\paren{\sin(\theta)\vl\vy-\lambda\paren{\cos(\theta)\hat{\ind}_N +\sin(\theta)\vy}}.
\end{equation*}
Taking the inner product of the above equation with $\hat{\ind}_N$  and simplifying gives
\begin{equation}\label{eqn:theta_evol}
\dot{\theta}=-\sin(2\theta)(L\vy,\vy).
\end{equation}
Using the equation for $\theta$, we can also determine that
\begin{equation}\label{eqn:y_evol}
\dot{\vy}=-2(\vl\vec{y}-\dpair{\vl\vy}{\vy}\vy).
\end{equation}
We define $e(t) = \dpair{\vl(t)\vy(t)}{\vy(t)}$ and we note that $e(t)$ takes values in the \emph{numerical range} of $\vl(t)$.
Put $\vm=\at{\vl}{W}$. Since $\vl$ is 
symmetric with real entries and since $\vy\in W$, we have the bound
\begin{equation}\label{eqn:spectral_bound}
\sigma_{min}(\vl(t))\leq \sigma_{min}(\vm(t))\leq e(t)\leq \sigma_{max}(\vm(t))\leq \sigma_{max}(\vl(t)).
\end{equation}
We shall occasionally make use of the following result which allows us to estimate the spectrum of a matrix in terms of its entries:

\begin{thm}[Gershgorin Disk Theorem]\label{thm:Gersh}
Let $A = (a_{ij})$ be a complex $n\times n$ matrix. Let $\displaystyle r_i =\sum_{j\neq i}\abs{a_{ij}}$ and let $D_i=\{z\in\C:\abs{z-a_{ii}}\leq r_i\}$. Then
\[\sigma(A)\subset \bigcup_{i=1}^{n}D_i.\]
\end{thm}
As an application we have

\begin{lem}\label{lem:unif_spec_bnd}
The spectrum $\sigma(L(t))$ is uniformly bounded in $t$ and satisfies
\begin{equation}
 \abs{\sigma(L(t))}\leq 4\abs{E}P^{\frac{1}{2}}.
\end{equation}

\end{lem}
\begin{proof}
That the spectrum is bounded is obvious since the entries of $\vl$ lie in a compact set and the determinant is a polynomial
and thus continuous in the entries. What we gain here is an explicit upper bound as follows.
Let $z\in\sigma(L)$ be a point in the spectrum. The Gershgorin theorem, Theorem  \ref{thm:Gersh}, implies that
\begin{align*}
\abs{z}&\leq \max_{i}\abs{\vl_{ii}}+\max_k\sum_{j\neq k}\abs{\vl_{kj}}\\
&\leq2\max_i\sum_{j\neq i}\abs{\gamma_{ij}}\\
&\leq 2\sum_{i}\sum_{j\neq i}\abs{\gamma_{ij}}\\
&\leq4\sum_{i<j}\abs{\gamma_{ij}}
\end{align*}
The Cauchy-Schwartz inequality and the constraint then implies that
\[\abs{z}\leq4\abs{E}^{\frac{1}{2}}\paren{\sum_{i<j}\abs{\gamma_{ij}}^2}^{\frac{1}{2}}\leq4\abs{E}P^{\frac{1}{2}},\]
which proves the lemma.
\end{proof}

\begin{remark}
 It is possible to modify the above argument and sharpen the above bound to 
\[
  \abs{\sigma(L(t))}\leq \sqrt{2N}(N-1)P^{\frac{1}{2}}.
\]
As we do not use this estimate, we will not prove it.
\end{remark}

Now let $\vec{u}\in W$ be a normalized eigenvector of $\sigma_{min}(\vm)$ i.e. $\vm\vec{u}=\sigma_{min}\vec{u}$ and for convenience
 let us put $\til{\sigma} =\sigma_{min}(\vm)$. A straightforward computation shows  that
\begin{equation}
\DO{t}\til{\sigma}=(\dot{\vm}\vec{u},\vec{u}),
\end{equation}
where $\dot{\vm}$ is the matrix obtained by differentiating the entries of $\vm$ with respect to $t$. The Gershgorin theorem, 
as in the proof of Lemma \ref{lem:unif_spec_bnd}, implies
\begin{align*}
\abs{\DO{t}\til{\sigma}}&\leq|(\dot{\vm}\vec{u},\vec{u})|\leq|\sigma(\dot{\vl})|\leq4\sum_{i<j}\abs{\dot{\gamma}_{ij}}.
\end{align*}
	
We note the following simple lemma which will allow us to simply the various expressions

\begin{lem}\label{lem:reduc}
 Let $\vx = \sqrt{R}\paren{\cos(\theta)\hat{\ind}_N +\sin(\theta)\vy}$ with $\vy\in W$ and $\norm{\vy}=1$. Then
\[
 \sum_{i< j}(y_i-y_j)^2=N,
\]
and thus
\[
\sum_{i< j}(x_i-x_j)^2=R\sin^2(\theta).
\]
 \end{lem}
 A direct substitution shows that
 \[\dot{\gamma}_{ij}=-\epsilon R\sin^2(\theta)\brac{(y_i-y_j)^2-\frac{(PN-Qe)+\gamma_{ij}(e-QN)}{\abs{E}(P-Q^2)}},\]
 and thus
  \[\abs{\dot{\gamma}_{ij}}\leq\epsilon R\sin^2(\theta)\brac{(y_i-y_j)^2+\frac{\abs{PN-Qe}+\abs{\gamma_{ij}}\abs{e-QN}}{\abs{E}(P-Q^2)}},\]
  which in turn implies, using the constraints, Lemma \ref{lem:reduc} and the Cauchy Schwartz inequality that
\begin{align*}
\sum_{i<j}\abs{\dot{\gamma}_{ij}}&\leq \epsilon R\sin^2(\theta)\brac{\sum_{i<j}(y_i-y_j)^2+\frac{\abs{PN-Qe}}{\abs{E}(P-Q^2)}\sum_{i<j}1+\frac{\abs{e-QN}}{\abs{E}(P-Q^2)}\sum_{i<j}\abs{\gamma_{ij}}}\\
&\leq\epsilon R\sin^2(\theta)\brac{N+\frac{\abs{PN-Qe}+P^{\frac{1}{2}}\abs{e-QN}}{(P-Q^2)}}.
\end{align*}
By Lemma \ref{lem:unif_spec_bnd}, $\abs{e}\leq 4\abs{E}P^{\frac{1}{2}}$ and thus the quantity in brackets in the last
 inequality above is seen to be uniformly bounded. Thus for some positive constant $K=K(R,N,P,Q)$ we have that
\begin{equation}
\abs{\DO{t}\til{\sigma}}\leq\epsilon K\sin^2(\theta)
\end{equation}
We are now ready to prove the local stability theorem.

\begin{proof}
We will construct a ``trapping region" for the flow by using the differential inequalities we have derived thus far. Equations  \eqref{eqn:theta_evol} and \eqref{eqn:spectral_bound} together imply that
\begin{equation}\label{eqn:trap_flow}
\left.\begin{aligned}
\DO{t}{\theta}&\leq-\til{\sigma}\sin(2\theta)\\
\abs{\DO{t}\til{\sigma}}&\leq\epsilon K\sin^2(\theta).
\end{aligned}\right\}
\end{equation}
Consider the set of curves through $(\theta_0,\sigma_0):=(\theta(0),\til{\sigma}(L(0)))$
\begin{align*}
\Gamma_1 = \{(\theta,\sigma)\,|\,\sigma^2+\epsilon K\ln\abs{\sec(\theta)} =\sigma_0^2+\epsilon K\ln\abs{\sec(\theta_0)}\}\\
\Gamma_2 = \{(\theta,\sigma)\,|\,\sigma^2-\epsilon K\ln\abs{\sec(\theta)} =\sigma_0^2-\epsilon K\ln\abs{\sec(\theta_0)}\}
\end{align*}
Direct computation shows that $\Gamma_1$ and $\Gamma_2$ intersects the positive  $\sigma$-axis at $\sqrt{\sigma_0^2+\epsilon K\ln\abs{\sec(\theta_0)}}$ and $\sqrt{\sigma_0^2-\epsilon K\ln\abs{\sec(\theta_0)}}$ respectively. If we set $\delta_{crit}=\arccos\paren{\exp\paren{-\sigma_0^2/\epsilon K}}$ it holds that $\epsilon K\ln\abs{\sec(\theta_0)}\leq\sigma_0^2$ for $|\theta_0|\leq\delta_{crit}$ so that the latter intersection is guaranteed to be real. Thus we may choose any $\delta\in[0,\delta_{crit}]$ and in particular for the choice
\begin{equation}
\delta=\arccos\paren{\exp\paren{-\tfrac{\sigma_0^2}{2\epsilon K}}},
\end{equation}
we see that uniformly for $\theta_0\in N_{\delta}(0)$, we have the containment\[ \brac{\sqrt{\sigma_0^2-\epsilon K\ln\abs{\sec(\theta_0)}},\sqrt{\sigma_0^2+\epsilon K\ln\abs{\sec(\theta_0)}}}\subset \brac{\frac{\sigma_0}{2},\frac{\sqrt{3}\sigma_0}{2}}.\]
Thus we see that by \eqref{eqn:trap_flow} the vector field is pointing \emph{downwards or tangential} on $\Gamma_1$ and \emph{upwards or tangential} on $\Gamma_2$. In the interior of the region bounded above by $\Gamma_1$, below by $\Gamma_2$, and to the right and left by $\theta=\pm\delta$ respectively, it is also pointing \emph{strictly leftwards} for $\theta>0$ and  \emph{strictly rightwards} for $\theta<0$. It then follows that $\displaystyle\lim_{t\to\infty}\theta(t) =0$ which proves the theorem.
\end{proof}
\end{document}